\documentclass[leqno,11pt]{amsart}
\usepackage{amsmath,amssymb,amsfonts,amsthm}
\usepackage{graphics}
\usepackage{graphicx,color}

\newtheorem{theorem}{Theorem}

\newtheorem{proposition}[theorem]{Proposition}
\newtheorem{remark}[theorem]{Remark}

\begin{document}
\title[ Pricing of Basket Options]{A Note on the Pricing of Basket Options Using Taylor Approximations }
\author[Olivares, Alvarez]{Pablo Olivares, Alexander Alvarez}
\address{Department of Mathematics, Ryerson University}

\begin{abstract}
In this paper we propose a closed-form approximation for the price of basket options under a multivariate Black-Scholes model, based on Taylor expansions and the calculation of mixed exponential-power moments of a Gaussian distribution. Our numerical results show that a second order expansion provides accurate prices of spread options with low computational costs, even for out-of-the-money contracts.
\end{abstract}

\maketitle

\section{Introduction}
The objective of the paper is the  pricing of basket options  using Taylor approximations  under diffusion multivariate models with constant covariance. Basket options are multivariate extensions of  European calls or puts. A basket option takes the weighted average of a group of $d$ stocks as the underlying, and produces a payoff equal to the maximum of zero and the difference between the weighted average and the strike (or the opposite difference for the case of a put). Index options, whose value depends on the movement of an equity or other financial index such as the S\&P500, are examples of basket options.\\
For the particular case of spread options, several approximations have been previously considered in the works of Kirk(1995), Carmona and Durrleman(2003), Li, Deng  and Zhou(2008, 2010), Venkatramanan and Alexander(2011)  where different ad-hoc approaches are studied. \\
As an alternative  Fast Fourier Transform methods have been successfully  implemented to compute spread prices under more general Levy processes, see  Hurd and Zhou(2009) and Cane and Olivares(2014) and under stochastic volatility models in Depmster and Hong(2000).\\
The approach to pricing by Taylor expansions can be traced back to  Hull and White(1987), where the price of a one dimensional derivative is calculated. On the other hand, following an idea in Pearson(1995) it can be extended to multidimensional contracts by conditioning on the remaining  $d-1$ underlying, reducing the problem to  one dimensional pricing with parameters arising from the resulting conditional distribution. It should be noticed that this technique has been used in Li, Deng  and Zhou(2008) for the case of spread options. Furthermore, in Li, Deng  and Zhou(2010) the Taylor expansion approximation is compared with other pricing techniques, proving to be effective and accurate for most values in the parametric space. \\
 Although in the same spirit,  in our case the expansion is done on the function resulting from the conditional price, as opposed to a development based on the conditional strike price, as previously considered by the authors cited above. Moreover, our method  hinges on the calculation of mixed exponential-power moments of a Gaussian distribution, it is extended to expansions about any point and higher dimensions. Our point of view may allow for a better control on the approximation, particularly for out-of-the-money options. On a related paper, see Alvarez, Escobar and Olivares(2011), we  apply a similar technique to the price of a spread option when correlation is stochastic, by expanding on the correlation matrix. \\
 The organization of the paper is the following, in section 2 we introduce some notations, the model and derive the Taylor approximation for basket options. In section 3 we specialize the formula for spread options and compute the mixed exponential-power moments of a Gaussian law. In section 4 we discuss our numerical results.
\section{Basket Derivatives and Taylor expansions}
We introduce some notations. Let  $(\Omega ,\mathcal{F},\{\mathcal{F}_{t}\}_{t>0},\mathbb{P})$ be a filtered probability space. We define the filtration  $\mathcal{F}^{X_t}:= \sigma(X_s, 0 \leq s \leq t)$ as the $\sigma$-algebra generated by the random variables $\{X_s, 0 \leq s \leq t \}$ completed in the usual way. Denote  by $\mathcal{Q}$ an equivalent martingale risk neutral measure and $E_{\mathcal{Q}}$ the expectation under $\mathcal{Q}$.\\
By $r$ we denote the (constant) interest rate, $A'$ represents the transpose of matrix $A=(a_{ij})_{1 \leq i,j \leq d}$ while $diag(A)$ is a vector with components $(a_{ii})_{1 \leq i \leq d}$. The d-dimensional column vector of ones is denoted by $1_d$.\\
For a $l$-times differentiable function $f$ in $\mathbb{R}^d$ and a vector $L=(l_1,l_2,\ldots,l_d)$ with $l_k \in \mathbb{N}$ such that
$\sum_{k=1}^n l_k=l$, $D^L f$ represents its mixed partial derivative of order $l$ differentiated $l_k$ times with respect to the variable $y_k$.\\
The process of spot prices is denoted by $S_t=\left(S^{(1)}_t,S^{(2)}_t,\ldots,S^{(d)}_t \right)'_{0 \leq t \leq T}$  and $Y_t=(Y_t^{(1)},Y_t^{(2)},\ldots,Y_t^{(d)})'_{0 \leq t \leq T}$ are the asset log-returns related by:
\begin{equation}\label{}
    S^{(j)}_t= S^{(j)}_0 \exp(Y_t^{(j)})\;\; \text{for}\;\; j=1,2,\ldots,d
\end{equation}
We analyze European Basket options whose payoff at maturity $T$, for a strike price $K$, is given by:
\begin{equation}\label{eq:basket}
    h(S_T)=\left( \sum_{j=1}^d w_j S_T^{(j)}-K \right)_+
\end{equation}
where $(w_j)_{1 \leq j \leq d}$ are some deterministic weights and $x_+=max(x,0)$. \\
As examples we have spread options, defined for  $d=2$ with  payoff:
\begin{equation}\label{eq:spread}
    h(S_T)=(  S_T^{(1)}- S_T^{(2)}-K )_+
\end{equation}
 Also, we have 3:2:1 crack spreads with $d=3$ and payoff:
    \begin{equation}\label{eq:cracks}
        h(S_T)=\left(\frac{2}{3} S_T^{(1)}-\frac{1}{3}S_T^{(2)}-S_T^{(3)}-K \right)_+
    \end{equation}
 where  $S_t^{(1)}$, $S_t^{(2)}$ and  $S_t^{(3)}$ are respectively the spot prices of gasoline, heating oil and crude oil.\\
 Exchange options are derivatives whose payoff is a particular case of (\ref{eq:spread}) when $K=0$. Exact formulas are available in the case of a diffusion, see Margrabe(1978).\\
We assume a multidimensional Black-Scholes dynamics  under the risk neutral probability following:
\begin{equation}\label{eq:bscmultid}
    dS_t=rS_t dt+ \Sigma^{\frac{1}{2}} S_t dB_t
\end{equation}
where $(B_t)_{t \geq 0}$ is a d-dimensional vector of  Brownian motions such that $d<B^{(l)}_t,B^{(m)}_t>=\rho_{lm} dt$, for $j,m=1,2,\ldots,d$ and $\Sigma $ is a positive definite symmetric matrix with components $(\sigma_{ij})_{i,j=1,2,\ldots,d}$ and $\sigma_{ii}=\sigma^2_i$.\\
We denote by $\tilde{Y}_t=(Y_t^{(2)},Y_t^{(3)}, \ldots, Y_t^{(d)})$ the  vector of log-returns, excluding the first component. The price of a basket option with maturity at $T>0$ and payoff $h(S_T)$ is:
\begin{eqnarray}\nonumber
p &=&e^{-rT}E_{\mathcal{Q}}h(S_{T})=E_{\mathcal{Q}}\left( e^{-rT}E_{\mathcal{Q}}\left[h(S_{T})|\mathcal{F}^{\tilde{Y}_T} \right]\right)= E_{\mathcal{Q}} \left[C(\tilde{Y}_T) \right] \\
 \label{eq:gen-price}
 &&
\end{eqnarray}
where:
\begin{equation*}
C (y):=E_{\mathcal{Q}}\left[ h(S_{T})|\mathcal{F}^{\tilde{Y}_T} \right]\mid_{\tilde{Y}_T=y}
\end{equation*}
Assuming $C(y)$ is smooth enough, we denote the n-th order Taylor development of $C$  around the point $y^* \in \mathbb{R}^{d-1}$ as  $\hat{C}_n(y)$. It is given by:
\begin{eqnarray}\label{eq:tqylorpricedev}
 \hat{C}_n(y)  &=& \sum_{l=0}^n \sum_{R_l} \frac{D^L C(y^*)}{l_1!l_2!\ldots l_{d-1}!} \prod_{k=1}^{d-1}  (y_k-y^*_k)^{l_k}
\end{eqnarray}
where:\\
$L=(l_1,l_2,\ldots,l_{d-1})$ and $R_l=\{L \in \mathbb{N}^{d-1}/ l_1+l_2+\ldots+l_{d-1}=l, \;\; 0 \leq l_k \leq l \}$.\\
The next proposition provides the Taylor approximation for the price $p$ of a basket option:
\begin{proposition}
The n-th order Taylor approximation around $y^*=(y^*_1,y^*_2, \ldots, y^*_{d-1})$ of the price $p$ of a basket option with payoff $h(S_T)$, defined as  $\hat{p}_n:= e^{-rT} E_{\mathcal{Q}}\hat{C}_n(\tilde{Y}_T)$,  under model (\ref{eq:bscmultid}), is given by:
\begin{equation}\label{eq:approxgen}
    \hat{p}_n= w_1 \sum_{l=0}^n \sum_{R_l} \frac{D^L C(y^*)}{l_1!l_2!\ldots l_{d-1}!}E_{\mathcal{Q}}\left [e^{-(r-\frac{1}{2}\sigma^2_{Y_T^{(1)}/\tilde{Y}_T})T+\mu_{Y_T^{(1)}/\tilde{Y}_T}}\prod_{k=1}^{d-1}  (Y_T^{(k+1)}-y^*_k)^{l_k} \right]
\end{equation}
where for $ y \in \mathbb{R}^{d-1}$:
\begin{equation}\label{eq:bschkvariab}
    C(y):=C_{BS}(K(y),\sigma_{Y_T^{(1)}/\tilde{Y}_T=y},S_0^{(1)})
\end{equation}
 is the Black-Scholes price of a call option  with strike price $K(y)$, maturity at $T>0$, volatility  $\sigma_{Y_T^{(1)}/\tilde{Y}_T=y^*}$, spot price $S_0^{(1)}$ and strike price:
  \begin{equation}\label{}
    K(y)=\frac{1}{w_1}e^{(r-\frac{1}{2}\sigma_{Y_T^{(1)}/\tilde{Y}_T=y}^2)T-\mu_{Y_T^{(1)}/\tilde{Y}_T=y}}\left(K-\sum_{j=2}^d w_j S_0^{(j)}e^{y^{*{(j)}}} \right)
  \end{equation}
 with:
 \begin{equation}\label{eq:mu}
\mu_{Y_T^{(1)}/\tilde{Y}_T}= (r-\frac{1}{2}\sigma_1^2)T+\Sigma_{1\tilde{Y}}\Sigma^{-1}_{\tilde{Y}}(\tilde{Y}-r+ \frac{1}{2}diag(\Sigma_{\tilde{Y}}))T
 \end{equation}
 \begin{equation}\label{eq:sigma}
\sigma_{Y_T^{(1)}/\tilde{Y}_T}=\sigma_1^2-\Sigma_{1\tilde{Y}}\Sigma^{-1}_{\tilde{Y}}\Sigma_{1\tilde{Y}}'
 \end{equation}
 \begin{equation*}
    \Sigma_{1\tilde{Y}}=(\sigma_{12},\sigma_{13},\ldots, \sigma_{1,d-1})'
 \end{equation*}
 and $\Sigma_{\tilde{Y}}$ is the covariance matrix of the vector $\tilde{Y}_T$.
\end{proposition}
\begin{proof}
From equation (\ref{eq:bscmultid}) a straightforward application of Ito formula leads to:
\begin{equation}\label{eq:mod:logretpri}
    Y_T= (r1_d-\frac{1}{2} diag(\Sigma))T+ \Sigma^{\frac{1}{2}} \sqrt{T} Z_d
\end{equation}
in law, where $Z_d$ is a random variable with a  multivariate normal distribution in $\mathbb{R}^d$ with zero mean and covariance matrix $I_d$. Hence $Y_T$ has also a multivariate normal distribution. Also conditionally on $\tilde{Y}_T$, the random variable $Y_T^{(1)}$ has a univariate normal distribution. Thus, we can write:
 \begin{equation}\label{eq:condnormalrv}
    Y_T^{(1)}= \mu_{Y_T^{(1)}/\tilde{Y}_T}+\sigma_{Y_T^{(1)}/\tilde{Y}_T}\sqrt{T}Z^{(1)}
 \end{equation}
 in law,  where $Z^{(1)}$ is independent of $Y_T$ and it has, conditionally on $\tilde{Y}_T$, a standard univariate normal distribution. Moreover it is well known, see for example Tong (1989), that  $\mu_{Y_T^{(1)}/\tilde{Y}_T}$ and $\sigma_{Y_T^{(1)}/\tilde{Y}_T}$ are given by equations (\ref{eq:mu}) and (\ref{eq:sigma}) respectively. \\
 Next, from equation (\ref{eq:gen-price}) we have:
\begin{eqnarray} \notag
p &=&e^{-rT}E_{\mathcal{Q}}\left( E_{\mathcal{Q}}\left( h(S_{T})|\mathcal{F}^{\tilde{Y}_T} \right) \right) \\ \notag
&=& w_1 e^{-rT}E_{\mathcal{Q}} \left(E_{\mathcal{Q}}\left[\left( S_0^{(1)}e^{Y_T^{(1)}}-\left(\frac{K}{w_1}-\sum_{j=2}^d \frac{w_j}{w_1}S_0^{(j)}e^{Y_T^{(j)}}\right)  \right)_+|\mathcal{F}^{\tilde{Y}_T} \right] \right) \\ \notag
&=& w_1 e^{-rT}E_{\mathcal{Q}} \left(E_{\mathcal{Q}} \left[\left( S_0^{(1)}e^{Y_T^{(1)}}-K'(\tilde{Y}_T)  \right)_+|\mathcal{F}^{\tilde{Y}_T} \right] \right) \\ \label{eq:pgral}
 \end{eqnarray}
 where $K'(y)=\frac{K}{w_1}-\sum_{j=2}^d \frac{w_j}{w_1}S_0^{(j)}e^{y^{(j)}}$.\\
 Moreover, substituting equation (\ref{eq:condnormalrv}) into (\ref{eq:pgral}) we have:
 \small{
 \begin{eqnarray*}
   p &=& w_1 e^{-rT}E_{\mathcal{Q}} \left[E_{\mathcal{Q}}\left(\left( S_0^{(1)}e^{\mu_{Y_T^{(1)}/\tilde{Y}_T}+\sigma_{Y_T^{(1)}/\tilde{Y}_T}\sqrt{T}Z^{(1)}}-K'(\tilde{Y}_T)  \right)_+|\mathcal{F}^{\tilde{Y}_T}  \right) \right] \\
    &=& w_1 e^{-rT}E_{\mathcal{Q}}\left[e^{-rT+\frac{1}{2}\sigma^2_{Y_T^{(1)}/\tilde{Y}_T}T+\mu_{Y_T^{(1)}/\tilde{Y}_T}}\right.\\
     &&\left. E_{\mathcal{Q}}\left( \left( S_0^{(1)}e^{rT-\frac{1}{2}\sigma^2_{Y_T^{(1)}/\tilde{Y}_T}T+\sigma_{Y_T^{(1)}/\tilde{Y}_T}\sqrt{T}Z^{(1)}}
 -  K'(\tilde{Y}_T) \right)_+|\mathcal{F}^{\tilde{Y}_T} \right)\right] \\
 &=& w_1 E_{\mathcal{Q}}\left[e^{-rT+\frac{1}{2}\sigma^2_{Y_T^{(1)}/\tilde{Y}_T}T+\mu_{Y_T^{(1)}/\tilde{Y}_T}} C(\tilde{Y}_T) \right]
  \end{eqnarray*}
  }
  where:
  \begin{eqnarray*}
  C(\tilde{Y}_T)&:=& C_{BS}(K(\tilde{Y}_T),\sigma_{Y_T^{(1)}/\tilde{Y}_T},S_0^{(1)})\\
  &=& e^{-rT} E_{\mathcal{Q}}\left[\left( S_0^{(1)}e^{(r-\frac{1}{2}\sigma_{Y_T^{(1)}/\tilde{Y}_T}^2)T+\sigma_{Y_T^{(1)}/\tilde{Y}_T}\sqrt{T}Z^{(1)}}
 -  K(\tilde{Y}_T) \right)_+|\mathcal{F}^{\tilde{Y}_T}\right]
  \end{eqnarray*}
  Applying a n-th order Taylor development around $y^*=(y^*_1,y^*_2, \ldots, y^*_{d-1}) \in \mathbb{R}^{d-1}$ to $C(y)$ we compute the approximated conditional price based on the first underlying and conditional on the remaining $d-1$ by :
\begin{equation}\label{eq:blacksholespriceexp}
 \hat{C}_{n}(\tilde{Y}_T)  = \sum_{l=0}^n \sum_{R_l} \frac{D^{L}C(y^*)}{\prod_{k=1}^{d-1} l_k!} \prod_{k=1}^{d-1}  (Y_T^{(k+1)}-y^*_k)^{l_k}
\end{equation}
  After replacing equation (\ref{eq:blacksholespriceexp}) into the expression for $p$ above  we get immediately equation (\ref{eq:approxgen}) in Proposition 1.
  \end{proof}
  \begin{remark}
  Notice that the approximation  $\hat{p}_k$ depends only on the derivatives of the function $C(y)$ with respect $y$, which in turn is computed as the Black-Scholes price composed with the function $K(y)$  and the mixed exponential-power moments of a Gaussian multivariate distribution.
  \end{remark}
  \begin{remark}
  Sensitivities to the parameters can be computed by a similar approximation, as \textit{Greeks} for a Black-Scholes option model are known. For example the delta with respect to the j-th asset can be approximated by:
\begin{equation*}\label{eq:delapproxgen}
    \hat{\Delta}^{(j)}_n= w_1 \sum_{l=0}^n \sum_{R_l} \frac{D^L \frac{\partial C(y^*)}{\partial s^{(j)}}}{l_1!l_2!\ldots l_{d-1}!}E_{\mathcal{Q}}\left [e^{-(r-\frac{1}{2}\sigma^2_{Y_T^{(1)}/\tilde{Y}_T})T+\mu_{Y_T^{(1)}/\tilde{Y}_T}}\prod_{k=1}^{d-1}  (Y_T^{(k+1)}-y^*_k)^{l_k} \right]
\end{equation*}
  \end{remark}
    \section{Pricing spreads options by Taylor approximations}
In order to illustrate the method studied in the previous section we consider the case of a bidimensional spread option  under model (\ref{eq:bscmultid}) with covariance matrix:
 \begin{equation*}
  \Sigma=  \left(
      \begin{array}{cc}
        \sigma_1^2 & 0 \\
        0 & \sigma_2^2 \\
      \end{array}
    \right)
 \end{equation*}
 We find the n-th Taylor approximation in this specific situation.
 Denoting by $d<B_t^{(1)},B_t^{(2)}>\rho dt$ we have that:
 \begin{equation}  \label{eq:risk-neutral-constant}
Y_T=(Y_T^{(1)},Y_T^{(2)}) \sim N \left( (r1_2-\frac{1}{2}  diag(\Sigma))T, T
\Sigma_{\rho} \right)
\end{equation}
where:
\begin{equation}\label{}
    \Sigma_{\rho}=\left(\begin{array}{ll}
                   \sigma_1^2 & \rho \sigma_1 \sigma_2 \\
                 \rho \sigma_1 \sigma_2   & \sigma_2^2
                 \end{array}
                 \right)
\end{equation}
From equation (\ref{eq:mod:logretpri}) the conditional distribution of $Y_T^{(1)}$ given $Y_T^{(2)}$ is:
 \begin{equation*}  \label{eq:risk-neutral-constant2}
Y_T^{(1)}/Y_T^{(2)} \sim N \left(  r(1-\frac{\sigma_1}{\sigma_2} \rho)T+ \frac{1}{2}\sigma_1 \sigma_2 \rho T+\frac{\sigma_1}{\sigma_2} \rho Y_T^{(2)}-\frac{1}{2} \sigma_1^2 T,  (1-\rho^2)\sigma_1^2 T \right)
\end{equation*}
Thus we can write:
\begin{equation*}\label{}
Y_T^{(1)}= \mu(Y^{(2)}_T)+\sigma \sqrt{T}Z
\end{equation*}
in law, where $Z \sim N(0,1)$ independent of $Y_T$, with
\begin{equation}\label{eq:mu2}
    \mu(Y_T^{(2)}):=\mu_{Y_T^{(1)}/\tilde{Y}_T}=r(1-\frac{\sigma_1}{\sigma_2}\rho) T+\frac{1}{2}\sigma_1( \sigma_2 \rho- \sigma_1) T+\frac{\sigma_1}{\sigma_2}\rho Y_T^{(2)}
\end{equation}
and
\begin{equation*}\label{eq:sigma2}
\sigma:=\sigma_{Y_T^{(1)}/\tilde{Y}_T}=\sqrt{(1-\rho^2)}\sigma_1
\end{equation*}
From Proposition 1 the n-th approximation simplifies to:
\begin{equation}\label{eq:approxsoread}
    \hat{p}_n=  \sum_{l=0}^n  \frac{D^lC(y^*)}{l!}E_{\mathcal{Q}}\left [e^{-(r-\frac{1}{2}\sigma^2)T+\mu(Y_T^{(2)})}  (Y_T^{(2)}-y^*)^{l} \right]
\end{equation}
Moreover:
\begin{equation*}
 E_{\mathcal{Q}}\left [e^{(-r+\frac{1}{2}\sigma^2)T+\mu(Y_T^{(2)})}  (Y_T^{(2)}-y^*)^{l} \right]= e^{A} E_{\mathcal{Q}}\left [ e^{\frac{\sigma_1}{\sigma_2}\rho Y_T^{(2)} }(Y_T^{(2)}-y^*)^{l}\right]
\end{equation*}
where:
\begin{eqnarray*}
    A&=& (-(r-\frac{1}{2}\sigma^2) +r (1-\frac{\sigma_1}{\sigma_2}) \rho  +\frac{1}{2}\sigma_1(\sigma_2 \rho-\sigma_1))T \\
        &=&-(\frac{1}{2}\rho^2 \sigma_1^2+r \frac{\sigma_1}{\sigma_2} \rho- \frac{1}{2} \sigma_1 \sigma_2 \rho)T
\end{eqnarray*}
Now, from equation (\ref{eq:risk-neutral-constant}) we have that $Y_T^{(2)} \sim N((r-\frac{1}{2}\sigma_2^2)T, T \sigma^2_2)$, then
 the exponential-power moments can be calculated as follows:
\begin{eqnarray*}
  && E_{\mathcal{Q}}\left [ e^{\frac{\sigma_1}{\sigma_2}\rho Y_T^{(2)} }(Y_T^{(2)}-y^*)^l\right] \\
  &=& \sum_{m=0}^l \left(\begin{array}{c}
                    l \\
                    m
                  \end{array} \right) \left((r-\frac{1}{2}\sigma_2^2)T-y^* \right)^{l-m}E_{\mathcal{Q}} \left[e^{\frac{\sigma_1}{\sigma_2}\rho Y_T^{(2)}}(Y_T^{(2)}-E_{\mathcal{Q}}(Y_T^{(2)}))^m \right]\\
   &=&
  \sum_{m=0}^l \left(\begin{array}{c}
                    l \\
                    m
                  \end{array} \right) \left((r-\frac{1}{2}\sigma_2^2)T-y^*\right)^{l-m}T^{\frac{m}{2}}\sigma_2^m  e^{\frac{\sigma_1}{\sigma_2}\rho (r-\frac{1}{2}\sigma_2^2)T} E_{\mathcal{Q}}\left[ e^{\sqrt{T}\sigma_1 \rho Z}Z^m \right]\\
                   &=&  e^{\frac{\sigma_1}{\sigma_2}\rho (r-\frac{1}{2}\sigma_2^2)T}
  \sum_{m=0}^l \left(\begin{array}{c}
                    l \\
                    m
                  \end{array} \right)  \left(\sqrt{T}\sigma_2 \right)^m B(y^*)^{l-m}  E_{\mathcal{Q}}\left[ e^{\sqrt{T}\sigma_1 \rho Z}Z^m \right]
\end{eqnarray*}
where:\\
$B(y^*)=(r-\frac{1}{2}\sigma_2^2)T-y^*$\\
Next  integrate by parts:
\begin{eqnarray*}
 E_{\mathcal{Q}}\left[ e^{\sqrt{T}\sigma_1 \rho Z}Z^m \right]  &=& \frac{1}{\sqrt{2 \pi}}\int_{\mathbb{R}}e^{-\frac{1}{2}(x^2-2 \sigma_1 \rho \sqrt{T}x)}x^m dx\\
    &=& e^{\frac{\sigma_1^2 \rho^2 T}{2} }\frac{1}{\sqrt{2 \pi}}\int_{\mathbb{R}}e^{-\frac{1}{2}(x- \sigma_1 \rho \sqrt{T})^2}x^m dx\\
   &=& e^{\frac{\sigma_1^2 \rho^2 T}{2} }\frac{1}{\sqrt{2 \pi}}\int_{\mathbb{R}}e^{-\frac{1}{2}y^2}(y+\sigma_1 \rho \sqrt{T})^m dy\\
&=& e^{\frac{\sigma_1^2 \rho^2 T}{2} } \sum_{\nu=0}^m \left(\begin{array}{c}
                    m \\
                    \nu
                  \end{array} \right)(\sigma_1 \rho \sqrt{T})^{m-\nu}E(Z^{\nu})\\
&=& e^{\frac{\sigma_1^2 \rho^2 T}{2} } \sum_{\nu=0}^{[\frac{m}{2}]} \left(\begin{array}{c}
                    m \\
                   2\nu
                  \end{array} \right)(\sigma_1 \rho \sqrt{T})^{m-2\nu}\frac{1}{\sqrt{2 \pi}}\int_{\mathbb{R}}e^{-\frac{1}{2}y^2}y^{2\nu} dy\\
&=& e^{\frac{\sigma_1^2 \rho^2 T}{2} } \sum_{\nu=0}^{[\frac{m}{2}]} \left(\begin{array}{c}
                    m \\
                   2\nu
                  \end{array} \right)(\sigma_1 \rho \sqrt{T})^{m-2\nu}(2 \nu-1)!!
\end{eqnarray*}
where $n!!$ is the double factorial defined as the product of all odd  numbers between 1 and $n$ including both. When the set is empty, by convention, the product is equal to one.\\
Similarly for $y^*=E_{\mathcal{Q}}(Y_T^{(2)})$ we  have:
\begin{eqnarray*}\label{}
 E_{\mathcal{Q}}\left [ e^{\frac{\sigma_1}{\sigma_2}\rho Y_T^{(2)} }(Y_T^{(2)}-y^*)^l\right] &=&  T^{\frac{l}{2}} \sigma_2^l e^{\frac{\sigma_1}{\sigma_2}\rho(r-\frac{1}{2}\sigma_2^2)T} e^{\frac{\sigma_1^2 \rho^2 T}{2} } \sum_{\nu=0}^l \left(\begin{array}{c}
                    l \\
                    \nu
                  \end{array} \right)(\sigma_1 \rho \sqrt{T})^{l-\nu}E(Z^{\nu})\\
 &=& T^{\frac{l}{2}} \sigma_2^l  e^{-A} \sum_{\nu=0}^{[\frac{l}{2}]} \left(\begin{array}{c}
                    l \\
                   2\nu
                  \end{array} \right)(\sigma_1 \rho \sqrt{T})^{l-2\nu}(2 \nu-1)!!
\end{eqnarray*}
After gathering all pieces  and substituting in equation (\ref{eq:approxsoread}) we have the following result:
\begin{proposition}
The n-th Taylor approximation of a spread contract with maturity at $T$ and strike price $K$, under the model (\ref{eq:bscmultid}) is given by:
\begin{eqnarray}\label{eq:approxspread2}\notag
    \hat{p}_n&=&   \sum_{l=0}^n  \sum_{m=0}^l \frac{ D^{l}C(y^*)}{l!}  \left(\begin{array}{c}
                    l \\
                    m
                  \end{array} \right)  \left(\sqrt{T}\sigma_2 \right)^m B(y^*)^{l-m}  E(m)
\end{eqnarray}
with:
\begin{equation*}\label{}
    E(m)=\sum_{\nu=0}^{m} \left(\begin{array}{c}
                    m \\
                   \nu
                  \end{array} \right)(\sigma_1 \rho \sqrt{T})^{m-\nu}E_{\mathcal{Q}} (Z^{\nu})
\end{equation*}
for $m=1,2,\ldots,k$ and $E(0)=1$, where $E_{\mathcal{Q}} Z^{\nu}=(\nu-1)!!$ if $\nu$ is even or zero if it is odd, and\\
\begin{equation*}
    K(y)=e^{(r-\frac{1}{2} \sigma^2) T-\mu(y)}(K+S_0^{(2)}e^y)=e^{-A} \left(Ke^{-\frac{\sigma_1}{\sigma_2}\rho y}+S_0^{(2)}e^{(1-\frac{\sigma_1}{\sigma_2}\rho) y} \right)
\end{equation*}
 with $\mu(y)$ given by equation (\ref{eq:mu2} ).
\end{proposition}
Next, we compute the derivatives of the function $C(y)$ with respect to $y$. From the Black-Scholes pricing formula:
\begin{equation*}\label{}
    C(y):= C_{BS}(K(y),\sigma,S^{(1)}_0)=S^{(1)}_0 N(d_1(K(y))-K(y)e^{-rT} N(d_2(K(y))
\end{equation*}
where:
\begin{eqnarray*}\label{}
    d_1(K(y))&=&\frac{\log \left(\frac{S^{(1)}_0}{K(y)} \right)+(r+\frac{\sigma^2}{2})T}{\sigma \sqrt{T}}\\
    d_2(K(y))&=& d_1(K(y)-\sigma \sqrt{T}
\end{eqnarray*}
and $N(.)$ is the cumulated distribution function of a standard normal distribution.\\
The first two derivatives are computed by elementary methods.\\
First notice that:
\begin{equation*}
   D^1 K(y)=e^{-A} \left(-\frac{\sigma_1}{\sigma_2} \rho K e^{-\frac{\sigma_1}{\sigma_2}\rho y}+S_0^{(2)}(1-\frac{\sigma_1}{\sigma_2} \rho)e^{(1-\frac{\sigma_1}{\sigma_2}\rho y)} \right)
\end{equation*}
 \begin{equation*}
   D^2 K(y)=e^{-A} \left((\frac{\sigma_1}{\sigma_2} \rho)^2 K e^{-\frac{\sigma_1}{\sigma_2}\rho y}+S_0^{(2)}(1-\frac{\sigma_1}{\sigma_2} \rho)^2 e^{(1-\frac{\sigma_1}{\sigma_2}\rho) y} \right)
\end{equation*}
Also:
 \begin{eqnarray*}\label{}
 D^{1}C_{BS}(y)&=& S_0^{(1)}f_Z(d_1(K(y)))D^1 d_1(K(y))-e^{-rT} D^1 K(y)N(d_2(K(y)))\\
 &-& e^{-rT} K(y)f_Z(d_2(K(y))) D^1 d_1(K(y))\\
 &=&-\frac{D^{1}K(y)}{K(y)\sigma \sqrt{T}} A_2(y)
 \end{eqnarray*}
where $f_Z$ is the density function of a standard normal random variable and
\begin{equation*}
  A_2(y)=S_0^{(1)}f_Z(d_1(K(y)))+\sigma \sqrt{T}e^{-rT} K(y)N(d_2(K(y)))-e^{-rT}K(y)f_Z(d_2(K(y)))
\end{equation*}
Similarly the second derivative is obtained as:
\begin{eqnarray*}
     D^2C(y) &=&-\frac{1}{\sigma \sqrt{T}}\left[A_2(y) \frac{K(y)D^2K(y)-(D^1K(y))^2}{K^2(y)}+D^1 A_2(y)\frac{D^{1}K(y)}{K(y)}\right]
 \end{eqnarray*}
 with:
 \begin{eqnarray*}
    D^1 A_2(y)&=&-S_0^{(1)}f_Z(d_1(K(y)))d_1(K(y)) D^1 d_1(K(y))+\sigma \sqrt{T}e^{-rT}D^1 K(y)N(d_2(K(y)))\\
    &+& \sigma \sqrt{T}e^{-rT}K(y)f_Z(d_2(K(y)))D^1 d_1(K(y))\\
    &+& e^{-rT} f_Z(d_2(K(y)))D^1 d_2(K(y))d_2(K(y))K(y)\\
     &=&\frac{D^{1}K(y)}{K(y)\sigma \sqrt{T}} \left[S_0^{(1)}f_Z(d_1(K(y)))d_1(K(y))+\sigma^2 Te^{-rT}K(y)N(d_2(K(y)))\right.\\
    &-& \left. 2 \sigma \sqrt{T}e^{-rT}K(y)f_Z(d_2(K(y)))-e^{-rT}K(y)f_Z(d_2(K(y)))d_2(K(y)) \right]\\
   \end{eqnarray*}
 In particular when we develop around  $y_{mean}=E_{\mathbb{Q}}(Y_T^{(2)})=(r-\frac{1}{2} \sigma_2^2) T$ we have the first and second approximations given respectively by :
\begin{eqnarray*}\label{}
 \hat{p}_1&=&   C(y_{mean})+\sigma_1 \sigma_2 \rho T D^{1}C(y_{mean})\\
 \hat{p}_2 &=& \hat{p}_1+  \frac{1}{2} \left[T \sigma^2_2(1+\sigma_1^2 \rho^2 T) \right]D^{2}C(y_{mean})
\end{eqnarray*}
More generally expanding around $y^*$  we have the first two approximations denoted by $\hat{p}_1(y^*)$ and $\hat{p}_2(y^*)$ respectively and given by:
\begin{eqnarray*}\label{}
    \hat{p}_1(y^*)&=& C(y^*)+ D^{1}C(y^*)(B(y^*) +\sqrt{T} \sigma_2 E(1))\\
    &=& C(y^*)+ D^{1}C(y^*)(B(y^*)+ T \sigma_1 \sigma_2 \rho )\\
    \hat{p}_2(y^*)&=& \hat{p}_1(y^*) +\frac{1}{2}D^{2}C(y^*) \left[B^2(y^*)+2 T \sigma_1 \sigma_2 \rho B(y^*)+ T \sigma_2^2(1+ T \sigma_1^2 \rho^2) \right]
\end{eqnarray*}

\section{Pricing Spreads: numerical results}
We consider spread options in the following benchmark numerical set:\\
$S_0^{(1)}=100$, $S_0^{(2)}=96$, $\sigma_1=0.3$, $\sigma_2=0.1$, $\rho=-0.3$, $r=0.03$, $K=1$ and $T=1$.\\
In Figure \ref{fig1} the graph of the conditional price $C(y)$ given by equation (\ref{eq:tqylorpricedev}) is shown (blue line), together with the first and second order Taylor approximation around the mean, for the benchmark parameter set.\\
Notice that the first approximation underestimates the price. Not surprisingly the second approximation estimates the price fairly well for values close to the point $y_{mean}$ while is less accurate for values far from the mean. Although it seems a drawback of the method it does not constitutes a serious problem as values far from the mean are unfrequent, thus the error in calculating the outer expected value by the Taylor approximation is small.
\begin{figure}[h!]
  \includegraphics[width=12 cm, height=10 cm]{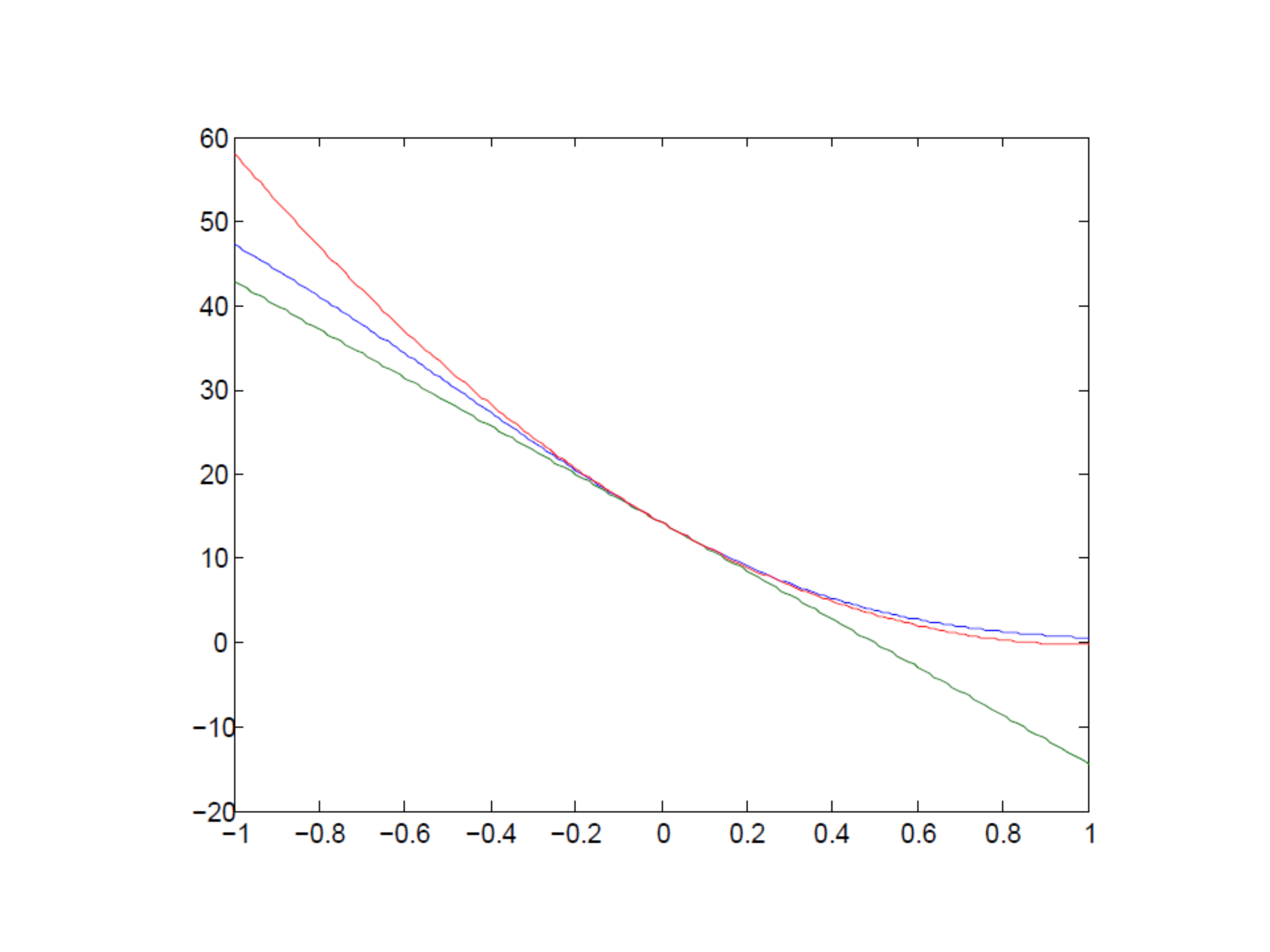}
  \caption{The function$C_{BS}(y)$ is shown in blue, together with the first and second order approximations around the mean for the benchmark parameters.}\label{fig1}
\end{figure}
In Figure \ref{fig2} a histogram for simulated returns on asset 1 (blue rectangles) and asset 2(red rectangles)  is shown. Notice that only a few values of the returns lie outside the interval $[-1,1]$.

\begin{figure}[h!]
  \includegraphics[width=12 cm, height=10 cm]{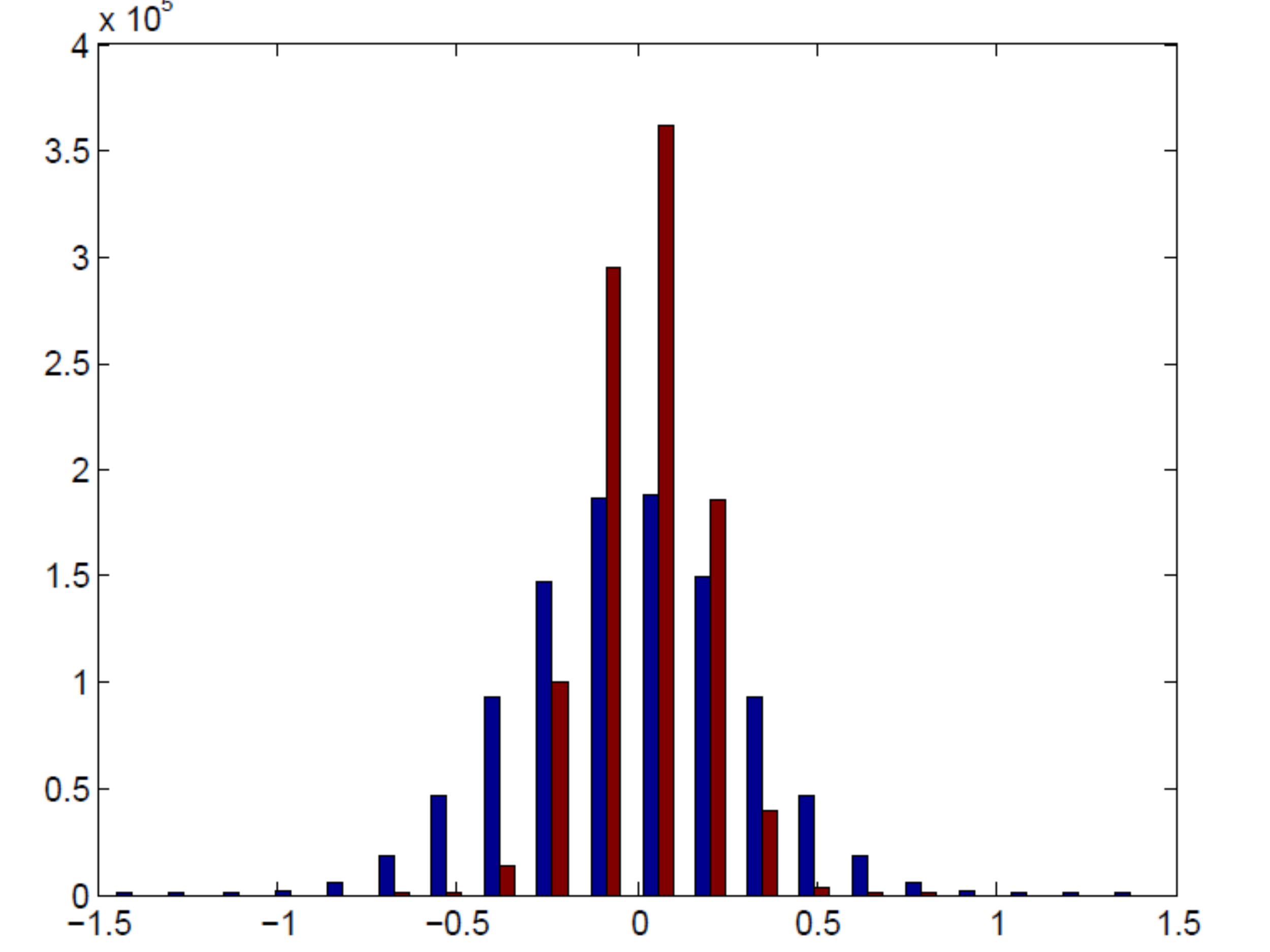}
  \caption{Histogram of simulated returns on  asset1 (blue rectangles) and asset 2 (red rectangles) for the benchmark parameter.}\label{fig2}
\end{figure}

Next we compare Taylor approximations with Monte Carlo simulations. In Table \ref{tab1}(column 2)  prices from Monte Carlo  are shown for the benchmark parameters, except the correlation parameter that takes values $\rho=-0.5,-0.3,0.3,0.5$. The number of simulations is $n=10^7$, where a stability of order $10^{-3}$ is attained. Partial Monte Carlo prices (shown in column 3) are obtained by sampling directly the one dimensional conditional price $C(Y_T^{(2)})$ and taking the corresponding  average of the payoff. It leads to a more efficient simulation algorithm as only one Brownian motion needs to be simulated, as oppose to two correlated Brownian in the standard Monte Carlo approach. It is done though at the expense of an extra evaluation of the Black-Scholes formula in every step.\\
Taylor prices of first and second order are shown in columns 4 and 5 of Table \ref{tab1}. The expansions take place around $y^*=0$. While in some cases the first order approximation reveals significant different with Monte Carlo, second order approximation shows an improved agreement with a relative error in the order of $10^{-4}$ for the parameter set considered.\\

\begin{table}
  \centering
  \begin{tabular}{|l|c|c|c|c|}
    \hline
Correlation & Monte Carlo & Partial Monte Carlo & First approx. & Second approx.\\ \hline
  $\rho =0.3000$ & 12.7843& 12.7907& 12.7889&12.7901 \\
 \hline
  $\rho =-0.3000$   &  14.9734&14.9826 &  13.6063&15.0065  \\
\hline
 $\rho =0.5000$    & 11.9525 & 11.9544 & 11.8085 &11.9646 \\
 \hline
  $\rho =-0.5000$    &15.6273 &  15.6302&  13.2767& 15.9238  \\
    \hline
  \end{tabular}
  \caption{Spread prices for the benchmark parameters and several values of $\rho$, using monte Carlo, partial Monte Carlo and first and second Taylor expansions around $y^*=0$.}\label{tab1}
\end{table}

For extreme values of the correlation coefficient $\rho$, e.g. larger  than an absolute value of $ 0.7$, the Taylor expansions around $y^*=0$ do not work well. Nevertheless it is interesting to notice that the approximations are rather sensible to the  point where the expansion is taken. Moreover, by slightly changing the latter the accuracy of the method can be considerably improved. In Table \ref{tab2} spread prices for the benchmark parameters and $\rho=-0.7$ for different expansion points are shown.

\begin{table}
  \centering
  \begin{tabular}{|c|c|c|c|c|}
    \hline
 Expansion point & Monte Carlo & Partial Monte Carlo & First approx. & Second approx.\\ \hline
   $y^*=-0.015$   &16.2463 &16.2540 &  12.3734 & 16.3011\\
    \hline
 $y^*=-0.02$ &16.2463 &16.2540 & 12.2966 &  15.8566\\
    \hline
  $y^*=-0.05$ & 16.2463 &16.2540 & 11.8434 & 12.9761 \\
    \hline
 $y^*=0$ &16.2463 & 16.2540 &12.5208 &  17.5217\\ \hline
   $y^*=0.01$  & 16.2463 &16.2540 &12.5089 & 18.2168\\
    \hline

  \end{tabular}
  \caption{Spread prices for the benchmark parameters, except $\rho=-0.7$  using Monte Carlo, partial Monte Carlo and first and second Taylor expansions  expanding around several values of $y^*$.}\label{tab2}
\end{table}
We test the Taylor expansion method for out-of-the-money contracts and compare with the price obtained via Monte Carlo with $n=10^7$ repetitions. The results are shown in Table \ref{outmoney}. The benchmark parameters are the same, except for the spot and strike prices that are changed accordingly. again a second order Taylor expansion seem to capture the Monte Carlo prices.
\begin{table}
 \centering
\begin{tabular}{|c|c|c|c|}
  \hline
  Parameters & Monte Carlo  & Taylor (first order) & Taylor (second order)  \\ \hline
$S_T^{(1)}=90, S_T^{(2)}=100$   & 7.040956 &  5.30281 & 7.0468998 \\
  $K=5, y^*=0.065$ &  &  &  \\ \hline
 $S_T^{(1)}=90, S_T^{(2)}=110$  &  4.8015937 &  3.442070 & 4.800319 \\
 $K=5, y^*=0.037$  &  &  &  \\ \hline
  $S_T^{(1)}=90, S_T^{(2)}=100$  & 5.7623 &  4.3248347  & 5.7726138 \\
 $K=10, y^*=0.05$  &  &  &  \\ \hline
  $S_T^{(1)}=90, S_T^{(2)}=110$  & 3.89825 & 2.71934  & 3.89966 \\
 $K=10, y^*=0.03$  &  &  &  \\ \hline
  \end{tabular}
  \caption{Prices of out-of-the-money spread contracts for selected strike and spot prices. Other parameters are kept within the benchmark set.   }\label{outmoney}
  \end{table}
 \section{Conclusions}
We present an efficient method to price basket options under a multidimensional Black-Scholes model, based on a Taylor expansion of the conditional one dimensional price resulting from fixing one of the underlying assets. The formula is given in terms of exponential-power moments  of a multivariate Gaussian law and the evaluation of certain derivatives in the Black-Scholes price.\\
We implement it numerically in the case of spread contracts. Within the benchmark parametric set this approach is  in closed agreement with the price obtained via Monte Carlo, even for deep out-of-the-money contracts, at considerable lesser computational effort. A second order development seems to be sufficient to achieve a relative error around $10^{-4}$.

\end{document}